\documentclass[12pt,english, russian]{article}
\usepackage{amsmath}
\usepackage{amsfonts}
\usepackage{amssymb}
\usepackage{amsthm}

\usepackage[left=25mm,top=26mm,right=26mm,bottom=30mm]{geometry}

\usepackage{enumerate}

\newtheorem{theorem}{Theorem}[section]

\newtheorem{lemma}{Lemma}[section]
\newtheorem{proposition}{Proposition}[section]

\title{Infinite Chain of Harmonic Oscillators Under the Action of the Stationary Stochastic Force}

\author{A.A.~Lykov\thanks{Mechanics and Mathematics Faculty, Lomonosov Moscow
  State University, Leninskie Gory~1, Moscow, 119991, Russia} \and M.V.~Melikian\footnotemark[1]}

\begin{document}
\maketitle

\begin{abstract}
We consider countable system of harmonic oscillators on the real line with quadratic interaction potential with finite support and local external force (stationary stochastic process) acting only on one fixed particle. In the case of positive definite potential and initial conditions lying in $l_2(\mathbb{Z})$-space the perpesentation of the deviations of the particles from their equilibrium points are found. Precisely, deviation of each particle could be represented as the sum of some stationary process (it is also time limiting process in distribution for that function) and the process which converges to zero as $t\rightarrow+\infty$ with probability one. The time limit for the mean energy of the whole system is found as well.
\end{abstract}

\allowdisplaybreaks 
\section{Introduction}

Systems of coupled harmonic oscillators and their generalizations are
the classical object of study in mathematical physics, generally speaking,
and in physics itself. The point is that in physics the harmonic oscillator model
plays an important role, especially in the study of small oscillations of systems
near stable equilibrium position. Vibrations of a load attached
to the spring (in a horizontal position) could serve as an example of such fluctuations in
classical mechanics, in quantum mechanics it can be vibrations of atoms in solids, molecules, etc. Existence of
solutions and their ergodic properties have been studied in \cite {LanfordLebowitz}.
There has also been extensive research on the convergence to equilibrium of the 
countable harmonic chain in contact with the thermostat
\cite {BoldrighiniPellegrinottiTriolo, DudKomechSpohn, Bogolyubov, SpohnLebowitz}.

In physics methods of study of many systems are of probabilistic nature:
typically, researchers describe how the movement of particles in a system affects
on the average picture of the behavior of the system as a whole.

Here one cannot fail to note the series of works \cite {DobrushinFritz, BoldrighiniPellegrinottiTriolo, BoldrighiniDobrushinSukhov, DobrushinPellegrinottiSuhovTriolo},
in which chains of harmonic oscillators with different
random initial conditions (this group of authors has more
early works, however, the above are seen as basic). Interest
to such models has not faded away so far, see, for example, \cite {BernardinHuveneersOlla}.
Here we would also like to note the works of Dudnikova T.V., for example, \cite {Dudnikova_2component, Dud2, DudKomechSpohn},
where the author studies the behavior of solutions at large times, deduces
variance estimates for them, proves the existence of wave operators
or the convergence of the distributions of solutions to some limiting measure.

Randomness can be introduced into the model in other ways. For instance,
in a number of works, where the heat flux in a finite disordered
chain of oscillators, the masses of particles can be assumed to be random.
Dyson \cite {Dyson} was the first to consider such a model. Later results
were obtained by Matsud and Ishii \cite {MatsudaIshii}, Leibovitz (see.
\cite {O'ConnorLebowitz, CasherLebowitz}). There are also models where
external influence is stochastic (see the article of Lykov A.A. \cite {Lykov_energy_growth}).

We would like to note the series of works of physicists (performed rather in mathematical
spirit) \cite {KuzkinKrivtsov_Energy, Krivtsov_Ballistic, Krivtsov_Babenkov_PM_2020, Gavrilov_CMT_2020},
where the authors investigate the propagation of heat along an infinite chain of
harmonic oscillators at the micro level to obtain a connection between
micro and macro descriptions (see also links inside), as well articles
\cite {Hemmen, Fox, FlorencioLee}.

\section{Model and main results}

In this paper, we consider large countable systems. The interaction between particles we consider only in the context of
Newton's classical mechanics (see \cite{LykovMalyshev_Nbody,LM_weak,LMChub,LM_mystery}).

More precisely, we consider a countable system of point particles with unit
masses on the real line $\mathbb{R}$ with coordinates $\{x_{k}\}_{k\in\mathbb{Z}}$
and velocities $\{v_{k}\}_{k\in\mathbb{Z}}$.
We define the formal Hamiltonian (total energy of the system) by
formula:
\[
H(x(t),v(t))=\sum_{k\in\mathbb{Z}}\frac{v_{k}^{2}}{2}+\sum_{k\in\mathbb{Z}}\frac{a_{kk}}{2}(x_{k}(t)-ka)^{2}+\sum_{\begin{array}{c}
k,j\in\mathbb{Z}\\
k\neq j
\end{array}}\frac{a_{kj}}{2}(x_{k}(t)-x_{j}(t)-(k-j)a)^{2},
\]
where parameters $a>0,\:a_{kk}\geqslant0$, and $(V)_{kj}=a_{kj}$ -- linear
operator in some linear space (conditions
will be discussed below). We call Hamiltonian ``formal''
due to the fact that, generally speaking, the question arises here
on the convergence of the series involved in its definition. In this case, the first sum
corresponds to the kinetic energy of the system and the remaining -- to the potential.
Namely, the second sum in the Hamiltonian means that the particle with number
$k$, where $a_{kk}>0$, is a harmonic oscillator (oscillation
occurs near the position $ka$), the last sum is responsible for
interaction between particles with numbers $k$ and $j$, where $a_{kj}\neq0$,
where, depending on the sign of $a_{kj}(x_{k}(t)-x_{j}(t)-(k-j)a)$,
there is attraction or repulsion between the corresponding particles.
The distance $(k-j)a$ here is the distance to which these
particles ``tend to''. Particle dynamics
is determined by the countable ODE system:
\begin{align}
\ddot{x}_{k}(t) & =-\frac{\partial H}{\partial x_{k}}=-a_{kk}(x_{k}(t)-ka)+\sum_{\begin{array}{c}
j\in\mathbb{Z}\\
j\neq k
\end{array}}a_{kj}(x_{k}(t)-x_{j}(t)-(k-j)a),\quad k\in\mathbb{Z},\label{mainEqX}
\end{align}
with initial conditions $x_{k}(0),v_{k}(0)$. Equilibrium position
(i.e. the particle configurance with minimum of energy) will be:
\[
x_{k}=ka,\quad v_{k}=0,\quad k\in\mathbb{Z}.
\]
This means that if particles at the initial moment are in this configuration
then the particles will not move at all, i.e. we will have $x_{k}(t)=ka,\ v_{k}(t)=0$
for all $t\geqslant0$. In this case, it will be convenient to move on to new
variables -- deviations:
\[
q_{k}(t)=x_{k}-ka,\quad p_{k}(t)=\dot{q}_{k}(t)=v_{k}(t).
\]

It is easy to see that new variables $q_{k}(t)$ satisfy the following
ODE system:
\begin{equation}
\ddot{q}_{k}=-a_{kk}q_{k}+\sum_{\begin{array}{c}
j\in\mathbb{Z}\\
j\neq k
\end{array}}a_{k,j}(q_{k}(t)-q_{j}(t)),\quad k\in\mathbb{Z}.\label{mainEqForQ}
\end{equation}

Therefore, we will further describe the systems of particles on a straight line by introducing the Hamiltonian depending  immediately on the deviations. Let us introduce the notation:
$$q(t)=\{q_j(t)\}_{j\in\mathbb{Z}}, \; p(t)=\{p_j(t)\}_{j\in\mathbb{Z}}.$$
Now consider a countable system of harmonic oscillators on a real line with
formal Hamiltonian:
\[
H(q(t),p(t))=\frac{1}{2}\sum_{k\in\mathbb{Z}}p_{k}^{2}(t)+\frac{1}{2}\sum_{k,j\in\mathbb{Z}}a(k-j)q_{k}(t)q_{j}(t),
\]
where $q_{j}(t),\,p_{j}(t)=\dot{q}_{j}(t)\in\mathbb{R}$ --- are the deviation
and the particle momentum (of particle number $j$) respectively, and the real-valued  function $a(k)$ satisfies three
conditions:
\begin{enumerate}
\item symmetry: $a(k)=a(-k)$; 
\item bounded support, i.e.: there exists $K\in\mathbb{N}$ such that for all $|k|>K$ holds $a(k)=0$; 
\item for all $\lambda\in\mathbb{R}$ holds: 
\begin{equation}
\omega^{2}(\lambda)=\sum_{k\in\mathbb{Z}}a(k)e^{ik\lambda}>0.\label{omega_2}
\end{equation}
\end{enumerate}

We assume also that the initial conditions  $\{q_{j}(0)\}_{j},\{p_{j}(0)\}_{j}$ lie in Hilbert space $L$: 
\[
L=\{\psi=(q,p):\ q\in l_{2}(\mathbb{Z}),\ p\in l_{2}(\mathbb{Z})\}.
\]

Let $V$ be the linear operator over $\mathbb{Z}$ corresponding to
$\{a(k)\}$, (i.e. $V_{k,j}=a(k-j)=a(j-k)$).

For all $\psi\in L$ we have the Fourier transform $\widehat{\psi}(\lambda)\in L^{2}([0;2\pi])$,
where $\widehat{\psi}(\lambda)=(\widehat{q}(\lambda),\widehat{p}(\lambda))$,
$\widehat{q}(\lambda)=\sum_{k\in\mathbb{Z}}q_{k}e^{ik\lambda}$. Then
$\widehat{Vq}(\lambda)=\omega^{2}(\lambda)\widehat{q}(\lambda)$.
Indeed,
\[
\widehat{Vq}(\lambda)=\sum_{j\in\mathbb{Z}}\sum_{k\in\mathbb{Z},\,|k-j|\leqslant K}a(k-j)q_{k}e^{ij\lambda}=\sum_{k\in\mathbb{Z}}\sum_{l\in\mathbb{Z},\,|l|\leqslant K}a(l)e^{i(k-l)\lambda}q_{k}=
\]
\[
=\sum_{k\in\mathbb{Z}}(\sum_{l\in\mathbb{Z},\,|l|\leqslant K}a(-l)e^{-il\lambda})e^{ik\lambda}q_{k}=\sum_{k\in\mathbb{Z}}(\sum_{l\in\mathbb{Z},\,|l|\leqslant K}a(l)e^{il\lambda})e^{ik\lambda}q_{k}=\omega^{2}(\lambda)\widehat{q}(\lambda).
\]
Further, the operator $V$ is positive definite:
\[
\sum_{k,j\in\mathbb{Z}}a(k-j)q_{k}q_{j}=(q,Vq)=(\widehat{q},\widehat{Vq})=(\widehat {q}(\lambda),\omega^{2}(\lambda)\widehat{q}(\lambda))=\omega^{2}(\lambda)(\widehat{q}(\lambda), \widehat{q}(\lambda))>0,
\]
for $q\neq0$ due to (\ref{omega_2}). Here in the second equality is used
that if $x,y\in l_{2}$ then $(x,y)=(\widehat{x},\widehat{y})$, where
$\widehat{x}$ is the Fourier transform of the element $x$, i.e. that
$(q_{1},q_{2})_{l_{2}(\mathbb{Z})}=\frac{1}{2\pi}\int_{0}^{2\pi} \widehat{q_{ 1}}(\lambda)\overline{\widehat{q_{2}}}(\lambda)d\lambda$.

Suppose, moreover, that on particle with a fixed number $n\in\mathbb{Z}$
external force $f(t)$ acts. Then the motion of the system is described by
the following infinite ODE system:
\begin{equation}
\ddot{q}_{j}=-\sum_{k}a(k-j)q_{k}+f(t)\delta_{j,n},\;j\in\mathbb{Z},\label{coordEq}
\end{equation}
where $\delta_{j,n}$ is the Kronecker symbol. We assume that $f(t)$ is a stochastic process satisfying the following condition:

A1) real-valued centered second-order stationary process
with continuous covariance function (see \cite{GrimmettStirzaker}, p. 361).

We say that sequences of stochastic processes $\{q_{k}(t)\}_{k\in\mathbb{Z}},\:\{p_{k}(t)\}_{k\in\mathbb{Z}}$
solve the system of equations (\ref{coordEq}) if they are continuously differentiable
in mean square and when substituted to (\ref{coordEq}), the right and left sides are equal
almost surely. More precisely, for any $j\in\mathbb{Z}$ and any $t\geqslant0$
equalities 
\[
\dot{q}_{j}=p_{j},
\]
\[
\dot{p}_{j}=-\sum_{k}a(k-j)q_{k}+f(t)\delta_{j,n},
\]
hold with probability one.

The following lemma on the existence and uniqueness of a solution of main system (\ref{coordEq}) holds.

\begin{proposition} \label{2_PropExistSol} Let the condition 
A1) be satisfied. Then for all $\psi(0)\in L$ there exists and is unique solution $\psi(t)=(q(t),p(t))$ of the system (\ref{coordEq}) with initial condition $\psi(0)$ such that $P(\psi(t)\in L)=1$ for all $t\geqslant0$. \end{proposition}
Uniqueness here means that if there is another solution
$\varphi(t)$ of the system (\ref{coordEq}) with inital condition $\psi$ such that $P(\varphi(t)\in L)=1$ for all $t\geqslant0$ then $\psi(t)$ and $\varphi(t)$ are stochastically equivalent, i.e. $P(\psi(t)=\varphi(t))=1$
for all $t\geqslant0$.

We are interested in the question of how the solution $\psi(t)$ and the average
energy of the system $\bold{E}H(\psi(t))$ as $t\rightarrow\infty$ behave. Before formulating the main results, we introduce one more condition on the external force.

Consider the set $E=\{\omega^{2}(\lambda):\ \lambda\in\mathbb{R}\}$ ($\omega^{2}(\lambda)$ is defined in (\ref{omega_2}))
--- range of the function $\omega^{2}(\lambda)$ (spectral
set of our system). Since $\omega^{2}(\lambda)$ is trigonometric
polynomial, then $E$ is a segment $[e_{1};e_{2}]$ of the real line.
Denote by $\mu(dx)$ the spectral measure of the process $f(t)$ and introduce the condition:

A2) the support of the spectral measure $\mu$ is isolated from the plus or minus ``of the root''
of the set $E$, i.e. there is an open set $U$ containing
$\pm[\sqrt{e_{1}},\sqrt{e_{2}}]$ such that $\mu(U)=0$.


\begin{theorem} \label{2_Theorem1}
Consider conditions A1) and A2) and $\psi(0)=(q(0),p(0))\in L$ hold. Then there is
random process $\eta(t)=(q^{\infty}(t),p^{\infty}(t))$ such that the following conditions hold:

\begin{enumerate} 

\item  $\eta(t)$ is a solution to the system (\ref{coordEq}) with some
initial conditions;

\item  the difference $\psi(t)-\eta(t)$ converges to zero as
$t\rightarrow+\infty$ component-by-component with probability one, and the trajectories of the process are continuous and infinitely differentiable a.s.;

\item  each component of $\eta(t)$ is a stationary process,
satisfying condition A1) and $P(\eta(t)\in L)=1$ for all $t\geqslant0$;

\item  there exist positive constants $c_{1},c_{2}$ and $0<r<1$
such that
\[
Dq_{k}^{\infty}(0)\leqslant c_{1}r^{|n-k|},\ Dp_{k}^{\infty}(0)\leqslant c_{2}r^{|n-k|}.
\]
\end{enumerate} \end{theorem} Thus, the process $\psi(t)$
is in some sense close to the stationary process $\eta(t)$. Besides this
in addition, $\eta(t)$ has ``nice'' properties. In particular,
the variance of the $\eta(0)$ components decreases exponentially with increasing distance
to the point of application of the external force.

Generally speaking, this assertion does not imply the weak
convergence of $\psi(t)$ components to the corresponding $\eta(0)$ components.
However, the following assertion can be proven:

\begin{theorem} \label{2_Theorem2} In addition to conditions A1)
and A2) assume that $f(t)$ is a strictly stationary process.
Then each component of $\psi(t)$ converges in distribution to the corresponding
component $\eta(0)$, i.e. for all $k\in\mathbb{Z}$
takes place the convergence
\[
q_{k}(t)\xrightarrow[]{d} q_{k}^{\infty}(0),\quad p_{k}(t )\xrightarrow[]{d} p_{k}^{\infty}(0)
\]
as $t\to\infty$. \end{theorem} Let us formulate theorems on mean energy of the
system.

\begin{theorem} \label{2_Theorem3} Let conditions A1)
and A2) be satisfied then
\[
\lim_{t\rightarrow+\infty}EH(\psi(t))=\alpha+H(\psi(0)),
\]
\[
EH(\eta(0))=\alpha,
\]
where we introduced the following constant
\[
\alpha=\frac{1}{4\pi}\int_{\mathbb{R}}\int_{0}^{2\pi}\frac{\omega^{2}(\lambda)+x^{2}}{(\omega^{2}(\lambda)-x^{2})^{2}}d\lambda\mu(dx).
\]
\end{theorem} Thus, the time limit for the average energy of the system
generally differs from the average energy of the limiting distribution,
which does not depend on the initial energy level, however, they will
coincide in the case of zero initial conditions.

\subsection{Proofs}

To begin with, let us note that it follows from the spectral theory that for $f(t)$
-- second-order stationary centered stochastic process:
\begin{equation}
B(s)=\int_{\mathbb{R}}e^{isx}\mu(dx),\;f(s)\overset{\text{a.s.}}{=}\int_{\mathbb{R }}e^{isx}Z(dx),\label{process}
\end{equation}
where $Z(dx)$ is an orthogonal measure, $\mu(dx)$ is a spectral measure,
and $B(s)$ is the covariance function.

\subsubsection{Proof of Proposition \ref{2_PropExistSol}}

Let us denote an operator in $L$:
\[
A=\left(\begin{array}{cc}
0&E\\
-V&0
\end{array}\right).
\]
Let's rewrite the system (\ref{coordEq}) in Hamiltonian form:
\begin{equation}
\begin{cases}
\dot{q_{j}}=p_{j},\\
\dot{p_{j}}=-\sum_{k}a(k-j)q_{k}+f(t)\delta_{j,n}.
\end{cases}\label{syst_main}
\end{equation}
Let's introduce the vector $\psi(t)=\left(\begin{array}{c}
q(t)\\
p(t)
\end{array}\right)$. Then the system will be rewritten in the form:
\begin{equation}
\dot{\psi}=A\psi+f(t)g,\label{eq_main}
\end{equation}
\[
g=(0,e_{n})^{T},\:0,e_{n}\in l_{2}(\mathbb{Z}),\:e_{n}(j)=\delta_{j,n}.
\]

The uniqueness of the solution follows from the linearity of the system. Indeed,
let $\psi(t)$ and $\varphi(t)$ be solutions of the system (\ref{eq_main})
with the same initial condition. Then $\epsilon(t)=\psi(t)-\varphi(t)$
is a solution to the homogeneous equation
\begin{equation}
\dot{\epsilon}=A\epsilon,\label{eq_odnor}
\end{equation}
with zero initial condition $\epsilon(0)=0$ and, moreover, $\epsilon(t)\in L$
almost certainly for all $t\geqslant0$. Thus, similarly to the
arguments of the classical theory of ODEs in Banach spaces we have
the required statement (see \cite{Dal_Krein_Stab}).

The solution of the system (\ref{eq_main}) can be expressed with the classical
formula for solving a nonhomogeneous ODE (see \cite{Dal_Krein_Stab}):
\[
\psi(t)=e^{At}(\psi(0)+\int_{0}^{t}e^{-As}gf(s)ds)=\psi_{0}(t)+\psi_{1}(t),
\]
where
\[
\psi_{0}(t)=e^{At}\psi(0),
\]
\[
\psi_{1}(t)=\int_{0}^{t}e^{A(t-s)}gf(s)ds.
\]
Note that since $V_{i,j}=a(i-j)$ and $a(k)$ has bounded
support the operator $A$ is a bounded linear operator
on $L$, hence the operator $e^{At}$ is well-defined
bounded operator on $L$. Hence $\psi_{0}(t)\in L$ for all
$t\geqslant0$.

Next, we turn to the consideration of $\psi_{1}(t)$. For this we need
lemma \begin{lemma} \label{2_lemmaEXP} For the operator $A$ the following is true:
\begin{equation}
e^{At}=\left(\begin{array}{cc}
\cos(\sqrt{V}t) & (\sqrt{V})^{-1}\sin(\sqrt{V}t)\\
-\sqrt{V}\sin(\sqrt{V}t) & \cos(\sqrt{V}t)
\end{array}\right),\label{exp_At}
\end{equation}
where the sine and cosine of the operator are defined by the corresponding series.
\end{lemma}

\begin{proof} Direct check or see \cite{Dal_Krein_Stab}.
\end{proof}

Let's return to the proof of the proposition \ref{2_PropExistSol}. Denote
$\psi_{1}(t)=(q^{(1)}(t),p^{(1)}(t))^{T}$. Then from the lemma \ref{2_lemmaEXP}
follows:
\begin{equation}
q_{k}^{(1)}(t)=\int_{0}^{t}f(s)S_{k,n}(t-s)ds,\;S(t)=(\sqrt{V})^{-1}\sin(\sqrt{V}t),\label{S_t}
\end{equation}
\begin{equation}
p_{k}^{(1)}(t)=\int_{0}^{t}f(s)C_{k,n}(t-s)ds,\;C(t)=\cos(\sqrt{V}t).\label{C_t}
\end{equation}
Note that consideration of the root of the operator is possible due to its positive
definiteness. It is necessary to prove that $\{q_{k}^{(1)}(t)\}_{k},\;\{p_{k}^{(1)}(t)\}_{k} \in l_{2}(\mathbb{Z})$
almost certainly. The proof is based on the lemma: \begin{lemma}
\label{2_lemmaBound_C_S} For all $t\geqslant0$ the following is true:
\[
\left|C_{k,n}\right|\leqslant\frac{v^{\rho}t^{2\rho}}{(2\rho)!}e^{\sqrt{v}t},\;\left|S_{k,n}\right|\leqslant\frac{v^{\rho}t^{2\rho+1}}{(2\rho+1)!}e^{\sqrt{v}t},
\]
where $S(t)$ and $C(t)$ are defined in $(\ref{S_t})-(\ref{C_t})$, $v=\left\Vert V\right\Vert _{l_{ 2}(\mathbb{Z})},\;\rho=\left\lceil \frac{|kn|}{K}\right\rceil $
(here $\left\lceil x\right\rceil $ is the smallest integer
not less than $x$). \end{lemma}

\begin{proof} See Lemma $3.2$ in \cite{Lykov_energy_growth} (p.
$7$). \end{proof} Let's continue the proof of the proposition.
\[
E|q_{k}^{(1)}(t)|^{2}=E\left(\int_{0}^{t}f(s)S_{k,n}(t-s)ds\overline{\int_{0}^{t}f(\tau)S_{k,n}(t-\tau)d\tau}\right)=
\]

\[
=\int_{0}^{t}\int_{0}^{t}E\left(\int_{\mathbb{R}}e^{isx}Z(dx)\int_{\mathbb{R}}e^{-i\tau y}\overline{Z}(dy)\right)S_{k,n}(t-s)\overline{S_{k,n}(t-\tau)}dsd\tau=
\]
\[
=\int_{0}^{t}\int_{0}^{t}\int_{\mathbb{R}}e^{i(s-\tau)x}\mu(dx)S_{k,n}(t-s)\overline{S_{k,n}(t-\tau)}dsd\tau=
\]
\[
=\int_{0}^{t}\int_{0}^{t}B(s-\tau)S_{k,n}(t-s)\overline{S_{k,n}(t-\tau)}dsd\tau\leqslant
\]
\[
\leqslant\sup_{s\in[0,t]}B(s)\int_{0}^{t}\int_{0}^{t}\frac{v^{\rho}(t-s)^{2\rho+1}}{(2\rho+1)!}e^{\sqrt{v}(t-s)}\frac{v^{\rho}(t-\tau)^{2\rho+1}}{(2\rho+1)!}e^{\sqrt{v}(t-\tau)}dsd\tau\leqslant
\]
\[
\leqslant\sup_{s\in[0,t]}B(s)\left(\frac{v^{\rho}t{}^{2\rho+2}}{(2\rho+2)!}e^{\sqrt{v}t}\right)^{2}.
\]
Let us show the correctness of the second equality. Possibility of permutation of integration
and taking the mathematical expectation follows from the continuity of the covariance
functions of the process $f(s)$ (from the condition) and the existence of integrals
\begin{equation}
\int_{0}^{t}f(s)S_{k,n}(t-s)ds\label{int_f_S}
\end{equation}
(in the mean square sense) and the Riemann integral
\begin{equation}
\int_{0}^{t}\int_{0}^{t}B(s-\tau)S_{k,n}(t-s)\overline{S_{k,n}(t-\tau)}dsd\tau.\label{B_S_S}
\end{equation}
The last integral exists due to continuity and boundedness of
integrands, and for the existence of the integral (\ref{int_f_S})
already described conditions are enough (see \cite{KramerLeadbetter} pp. $94-129$).

From here:
\[
\sum_{k}E|q_{k}^{(1)}(t)|^{2}\leqslant\sup_{s\in[0,t]}B(s)\sum_{\rho}\left(\frac{v^{\rho}t{}^{2\rho+2}}{(2\rho+2)!}e^{\sqrt{v}t}\right)^{2}\leqslant\sup_{s\in[0,t]}B(s)e^{2\sqrt{v}t}t^{2}\sum_{\rho}\frac{(vt^{2}){}^{2\rho}}{(2\rho)!}=
\]
\[
=\sup_{s\in[0,t]}B(s)e^{2\sqrt{v}t}t^{2}ch(vt^{2})<\infty,
\]
whence, by the corollary of Levy's monotone convergence theorem (see \cite{Kolmogorov_Fomin}
p. $306$) it follows that
\[
\sum_{k}|q_{k}^{(1)}(t)|^{2}<\infty\;a.s.,
\]
i.e. $\{q_{k}^{(1)}(t)\}_{k}\in l_{2}(\mathbb{Z})$ almost certainly.
Similarly, $\{p_{k}^{(1)}(t)\}_{k}\in l_{2}(\mathbb{Z})$ almost certainly.
The assertion has been completely proven.

\subsubsection{Proof of the Theorem \ref{2_Theorem1}}

Let's introduce $\eta(t)$ by the formula:
\[
\eta(t)=-\int_{\mathbb{R}}e^{itx}R_{A}(ix)Z(dx)g,
\]
where $R_{A}(z)=(A-zI)^{-1}$ is the resolvent of the operator $A$ (here
$I$ is the identity operator over $\mathbb{Z}\mathbb{\times Z}$), and we prove that the stochastic process introduced in this way satisfies all conditions of the theorem.
Moreover, the resolvent is bounded due to the conditions of Theorem \ref{2_Theorem1}
(namely, condition A$2)$ on page $13$), which implies the convergence of the considered
integral.

First, we prove that it is a solution to the system (\ref{syst_main}):
\[
-\dot{\eta}(t)+A\eta(t)+f(t)g=\int_{\mathbb{R}}ixe^{itx}R_{A}(ix)Z(dx)g-\int_{\mathbb{R}}e^{itx}AR_{A}(ix)Z(dx)g+f(t)g=
\]
\[
=\int_{\mathbb{R}}e^{itx}(ixI-A)R_{A}(ix)Z(dx)g+\int_{\mathbb{R}}e^{itx}Z(dx)g=
\]
\[
=\int_{\mathbb{R}}e^{itx}((ixI-A)(A-ixI)^{-1}+I)Z(dx)g=0.
\]

It is possible to introduce differentiation under the integral sign in view of the existence of
integral
\[
\int_{\mathbb{R}}\left(R_{A}(ix)g,R_{A}(ix)g\right)\mu(dx)
\]
(See \cite{Ventzel}, p. $94$).

\paragraph*{Proof of item 4.}

Since the limit vector is stationary, consider $\eta(0)$
and for brevity we denote it by $\xi=\eta(0)$. Since $E\eta(0)=0,$
then
\[
cov(\eta_{j}(0),\eta_{k}(0))=E\eta_{j}(0)\eta_{k}(0)=E\xi_{j}\xi_{k} .
\]

Further, in view of the fact that the relation
\[
E\left(\int fZ(dx)\overline{\int gZ(dx)}\right)=\int_{\mathbb{R}}f\overline{g}\mu(dx),
\]
is true we have
\begin{equation}
C\equiv E\xi\overline{\xi}^{T}=\int_{\mathbb{R}}e^{itx}R_{A}(ix)gg^{T}\overline{R_{A}(ix)^{T}}e^{-itx}\mu(dx)=\int_{\mathbb{R}}R_{A}(ix)gg^{T}\overline{R_{A}(ix)^{T}}\mu(dx).\label{C_def_int}
\end{equation}
Denote
\begin{equation}
C(x)=R_{A}(ix)gg^{T}\overline{R_{A}(ix)^{T}}.\label{C_x_def}
\end{equation}
Let us find the resolvent of the operator $A$:
\[
R_{A}(\lambda)=(A-\lambda I)^{-1},
\]
\[
\left(\begin{array}{cc}
-\lambda E & E\\
-V & -\lambda E
\end{array}\right)\left(\begin{array}{cc}
A & B\\
C & D
\end{array}\right)=\left(\begin{array}{cc}
E & 0\\
0 & E
\end{array}\right),
\]
\[
\left(\begin{array}{cc}
(-\lambda A+C) & (-\lambda B+D)\\
(-VA-\lambda C) & (-VB-\lambda D)
\end{array}\right)=\left(\begin{array}{cc}
E & 0\\
0 & E
\end{array}\right),
\]
from where
\[
R_{A}(\lambda)=\left(\begin{array}{cc}
-\lambda R_{V}(-\lambda^{2}) & -R_{V}(-\lambda^{2})\\
E-\lambda^{2}R_{V}(-\lambda^{2}) & -\lambda R_{V}(-\lambda^{2})
\end{array}\right),
\]
thus
\[
R_{A}(ix)=\left(\begin{array}{cc}
-ixR_{V}(x^{2}) & -R_{V}(x^{2})\\
E+x^{2}R_{V}(x^{2}) & -ixR_{V}(x^{2})
\end{array}\right),
\]
\begin{equation}
R_{A}(ix)g=\left(\begin{array}{c}
-R_{V}(x^{2})e_{n}\\
-ixR_{V}(x^{2})e_{n}
\end{array}\right),\label{R_A_ix_g}
\end{equation}
\[
g^{T}\overline{R_{A}(ix)^{T}}=\left(\begin{array}{cc}
- & e_{n}^{T}R_{V}(x^{2})\end{array},\,ixe_{n}^{T}R_{V}(x^{2})\right).
\]
For convenience of notation we denote
\[
\rho=R_{V}(x^{2}),\varGamma=e_{n}e_{n}^{T},
\]
then (\ref{C_x_def}) becomes:
\[
C(x)=\left(\begin{array}{cc}
\rho\varGamma\rho\;\; & -ix(\rho\varGamma\rho)\\
ix(\rho\varGamma\rho)\; & x^{2}(\rho\varGamma\rho)
\end{array}\right).
\]

In (\ref{C_def_int}) $\int_{\mathbb{R}}ix(\rho\varGamma\rho)\,\mu(dx)=0$,
since $\mu$ is a symmetric measure due to the realness of the process,
$xR_{V}(x^{2})$ is an odd function. Let us pass to the integral $\int_{\mathbb{R}}\rho\varGamma\rho\,\mu(dx)$.
Denote
\begin{equation}
c_{k,j}=(Ce_{k},e_{j})=\int_{\mathbb{R}}(\rho\varGamma\rho e_{k},e_{j})\mu(dx),\label{c_k_j}
\end{equation}
\begin{equation}
c_{k,j}(x)=(\rho\varGamma\rho e_{k},e_{j})=(\varGamma\rho e_{k},\rho e_{j}).\label{c_k_j_x}
\end{equation}
Note that
\[
\widehat{\varGamma x}(\lambda)=\sum_{j}(\varGamma x)_{j}e^{ij\lambda}=x_{n}e^{in\lambda},
\]
\[
\widehat{\rho e_{j}}=\frac{\widehat{e_{j}}}{\omega^{2}(\lambda)-x^{2}}=\frac{e^{ij\lambda}}{\omega^{2}(\lambda)-x^{2}}=e^{ij\lambda}b_{x}(\lambda),
\]
\[
\widehat{\varGamma\rho e_{k}}=(\rho e_{k})_{n}e^{in\lambda}=(\rho e_{k},e_{n})e^{in\lambda}=e^{in\lambda}(\widehat{\rho e_{k}},\widehat{e_{n}})=
\]
\[
=\frac{e^{in\lambda}}{2\pi}\int_{0}^{2\pi}\frac{e^{-in\nu}e^{ik\nu}d\nu}{\omega^{2}(\nu)-x^{2}},
\]
hence, in view of the self-adjointness of the operator $\rho$: 
\[
(\rho\varGamma\rho e_{k},e_{j})=(\varGamma\rho e_{k},\rho e_{j})=(\widehat{\varGamma\rho e_{k}},\widehat{\rho e_{j}})=
\]
\[
=\left(\frac{1}{2\pi}\right)^{2}\int_{0}^{2\pi}\frac{e^{i(n-j)\lambda}d\lambda}{\omega^{2}(\lambda)-x^{2}}\int_{0}^{2\pi}\frac{e^{i(k-n)\lambda}d\lambda}{\omega^{2}(\lambda)-x^{2}},
\]
and, taking into account the parity and periodicity of the integrands, we arrive at
to
\[
(\rho\varGamma\rho e_{k},e_{j})=\left(\frac{1}{2\pi}\right)^{2}\int_{0}^{2\pi}\frac{\cos((n-j)\lambda)d\lambda}{\omega^{2}(\lambda)-x^{2}}\int_{0}^{2\pi}\frac{\cos((k-n)\lambda)d\lambda}{\omega^{2}(\lambda)-x^{2}}.
\]
Consider the cases:

1. $j=k\neq n\rightarrow+\infty$. Here
\[
\int_{0}^{2\pi}b_{x}(\lambda)\cos((k-n)\lambda)d\lambda=\frac{1}{k}(b_{x}(\lambda)\sin((k-n)\lambda)|_{\lambda=0}^{\lambda=2\pi}-\int_{0}^{2\pi}b'_{x}(\lambda)\sin((k-n)\lambda)d\lambda)=
\]
\[
=\frac{1}{k^{2}}(b''_{x}(\lambda)\cos((k-n)\lambda)|_{\lambda=0}^{\lambda=2\pi}-\int_{0}^{2\pi}b''_{x}(\lambda)\cos((k-n)\lambda)d\lambda)=...,
\]
thus there is an exponentially fast decay (faster than
any degree, i.e. $\overline{o}\left(\frac{1}{(k-n)^{\infty}}\right)$).

2. In other cases, consider
\[
c_{k,j}=\int_{\mathbb{R}}c_{k,j}(x)\mu(dx),
\]
where $c_{k,j}(x)$ is entered in (\ref{c_k_j_x}). Denote:
\[
h_{k}(x)=\frac{1}{2\pi}\int_{0}^{2\pi}\frac{e^{i(n-k)\lambda}d\lambda}{\omega^{2}(\lambda)-x^{2}},
\]
We make a replacement $z=e^{i\lambda}$:

\begin{equation}
h_{k}(x)=\frac{1}{2\pi i}\ointop_{|z|=1}\frac{1}{P(z)-x^{2}}z^{n-k-1}dz,\label{h_k_x}
\end{equation}
where
\[
P(z)=\omega^{2}(\lambda)|_{\lambda=\lambda(z)},
\]
i.e.
\[
P(z)=\sum_{j}a(j)e^{ij\lambda}=a(0)+\sum_{j=1}^{K}a(j)(z^{j}+z^{-j}).
\]
We are interested in the roots of the equation
\[
P(z)-x^{2}\equiv a(0)+\sum_{j=1}^{K}a(j)(z^{j}+z^{-j})-x^{2}=0,
\]
\begin{equation}
a(K)z^{2K}+...+(-x^{2}+a(0))z^{K}+...+a(K)=0.\label{eq_P_z}
\end{equation}
Obviously, $P(z)=P(1/z),$ hence, if $z$ is a root, then $1/z$
is the root, so the zeros are invariant under the inversion of the unit
circle $|z|=1$. These values are the values from the spectrum,
which we ''avoid''. In total, we have $K$ inverse pairs, which
needs to be bypassed. You can choose a ring (neighborhood of the unit circle),
where the inverse function is holomorphic
\begin{equation}
g(z)=\frac{1}{P(z)-x^{2}}. \label{g(z)_def}
\end{equation}
To do this, we find the maximum modulo root of the equation (\ref{eq_P_z}),
lying inside the unit circle. We denote its modulus by $R(x)$. Then
its inverse pair has modulus equal to $\frac{1}{R(x)}$, moreover
this will be the smallest modulus of roots lying outside the unit circle.
Thus, $g(z)$ is holomorphic in the ring
\[
R(x)<|z|<\frac{1}{R(x)},\:0<R(x)<1<\frac{1}{R(x)}.
\]
We choose $\epsilon$ close to zero and denote by $\rho=\frac{1-\epsilon}{R(x)}$,
then the contour $|z|=\rho$ lies in the holomorphy ring $g(z)$, hence
it is possible in (\ref{h_k_x}) to replace the integration contour with the considered one,
then
\[
h_{k}(x)=\frac{1}{2\pi i}\ointop_{|z|=\rho}\frac{1}{P(z)-x^{2}}z^{n-k-1}dz,
\]
let's evaluate the module:
\[
|h_{k}(x)|\leqslant\frac{1}{2\pi}\ointop_{|z|=\rho}|g(z)z^{n-k-1}|\cdot|dz|=\frac{1}{2\pi}\ointop_{|z|=\rho}|g(z)|\rho^{n-k-1}\cdot|dz|,
\]
in the last integral we make the change $z=\rho e^{i\phi},$ then
\[
|h_{k}(x)|\leqslant\frac{\rho^{n-k-1}}{2\pi}\int_{0}^{2\pi}|g(\rho e^{i\phi})|\rho d\phi=\frac{\rho^{n-k}}{2\pi}\int_{0}^{2\pi}|g(\rho e^{i\phi})|d\phi=\frac{\rho^{n-k}}{2\pi}r(x).
\]
Note that since there exists $q$ such that for all $x$ from $\mathbb{R}\setminus E$
$\rho\geqslant1/q>1$ is true, whence $1/\rho\leqslant q<1$ and
\[
|h_{k}(x)|\leqslant\frac{q^{k-n}}{2\pi}r(x),
\]
\[
c_{k,k}=\sigma_{k}^{2}\leqslant\int_{\mathbb{R}}|h_{k}(x)|^{2}\mu(dx)\leqslant\frac{q^{2(k-n)}}{4\pi^{2}}\int_{\mathbb{R}}r^{2}(x)\mu(dx)\rightarrow0,\:k\rightarrow+\infty,
\]

if $\int r^{2}(x)\mu(dx)$ converges. Let's prove that this is indeed the case:
\[
r(x)=\frac{1}{2\pi}\int_{0}^{2\pi}|g(\rho e^{i\phi})|d\phi\leqslant\frac{1}{2\pi}\max_{z:|z|=\rho}|g(z)|\cdot2\pi=\max_{z:|z|=\rho}|g(z)|,
\]
where $g(z)$ is introduced in (\ref{g(z)_def}). Let us notice, that

\begin{equation}
|P(z)-x^{2}|\geqslant||P(z)|-x^{2}|\geqslant\min\{|I-x^{2}|,\;|S-x^{2}|\}>0,\label{P(z)_min_I_S}
\end{equation}
where
\[
I=I(\rho)=\inf_{z:|z|=\rho}|P(z)|,\;S=S(\rho)=\sup_{z:|z|=\rho}|P(z)|,
\]
and the last strict inequality holds due to the choice of radius $\rho$ of the circle.

Further, for $\rho=1$ the segment $[I(\rho),\;S(\rho)]$ coincides with the set
$E$. Since $I$ and $S$ are continuous functions of $\rho$, there exists
a neighborhood of the point $\rho=1$ such that for any $\rho$ from this
neighborhood segment $[I(\rho),\;S(\rho)]$ lies in $U$, where $U$ is defined
in condition A2. Then there is a positive constant $\epsilon>0$
such that for any $x\in\mathbb{R}\diagdown U$ we have the following inequality:
\[
|P(z)-x^{2}|\geqslant\epsilon>0,
\]
wherefrom
\[
|g(z)|\leqslant\frac{1}{\epsilon}.
\]
So, for all $z:|z|=\rho,$ and all $x\in\mathbb{R}\diagdown U$
we have an estimate:
\[
r(x)\leqslant\frac{1}{\epsilon}.
\]
For $x\rightarrow\infty$ the function $r(x)$ decreases as $\frac{1}{x^{2}}$
due to the estimate (\ref{P(z)_min_I_S}), whence the convergence of the integral follows.

\paragraph*{Next, let's move on to item 3.}

It follows from item $4$ that the initial conditions $\eta(0)\in L$. Indeed:

\[
\sum_{k\in\mathbb{Z}}Dq_{k}^{\infty}(0)\leqslant c_{1}\sum_{k\in\mathbb{Z}}r^{|n-k|}\leqslant c_{1}\sum_{k\in\mathbb{Z}}r^{|k|}<\infty,
\]
\[
\sum_{k\in\mathbb{Z}}Dp_{k}^{\infty}(0)\leqslant c_{2}\sum_{k\in\mathbb{Z}}r^{|n-k|}\leqslant c_{2}\sum_{k\in\mathbb{Z}}r^{|k|}<\infty,
\]
whence, by the corollary of Levy's monotone convergence theorem (see \cite{Kolmogorov_Fomin}
p. $306$) it follows that
\[
\sum_{k}|q_{k}^{\infty}(0)|^{2}<\infty\;a.s.,
\]
i.e. $\{q_{k}^{\infty}(0)\}_{k}\in l_{2}(\mathbb{Z})$ almost certainly.
Similarly, $\{p_{k}^{\infty}(0)\}_{k}\in l_{2}(\mathbb{Z})$ is almost
surely. Then the statement \ref{2_PropExistSol} implies what is required.

\paragraph*{Let's prove the second item.}

The difference $\psi(t)-\eta(t)$ at the initial time point lies in $L$
and is a solution of the homogeneous equation, hence, componentwise tends to zero
almost surely. Indeed, $\epsilon(t)=\psi(t)-\eta(t)$
is a solution to the homogeneous equation
\begin{equation}
\dot{\epsilon}=A\epsilon,\label{eq_odnor-1}
\end{equation}
with the initial condition $\epsilon(0)\in L$ almost surely for all $t\geqslant0$,
and has the form:
\[
\epsilon(t)=e^{At}\epsilon(0)=e^{At}\left(\begin{array}{c}
q(0)\\
p(0)
\end{array}\right), 
\]
which obviously implies continuity and infinite differentiability a.s. of process trajectories.
Consider one of the coordinates:
\[
\left(\epsilon(t),\left(\begin{array}{c}
e_{k}\\
0
\end{array}\right)\right)=\left(e^{At}\left(\begin{array}{c}
q(0)\\
p(0)
\end{array}\right),\left(\begin{array}{c}
e_{k}\\
0
\end{array}\right)\right)=\left(\widehat{e^{At}\left(\begin{array}{c}
q(0)\\
p(0)
\end{array}\right)},\widehat{\left(\begin{array}{c}
e_{k}\\
0
\end{array}\right)}\right)=^{(3)}
\]
\[
=\left(\left(\begin{array}{c}
\cos(\omega(\lambda)t)\widehat{q(0)}(\lambda)+\frac{\sin(\omega(\lambda)t)}{\omega(\lambda)}\widehat{p(0)}(\lambda)\\
-\omega(\lambda)\sin(\omega(\lambda)t)\widehat{q(0)}(\lambda)+\cos(\omega(\lambda)t)\widehat{p(0)}(\lambda)
\end{array}\right),\left(\begin{array}{c}
e^{ik\lambda}\\
0
\end{array}\right)\right)=
\]
\[
=\frac{1}{2\pi}\int_{0}^{2\pi}e^{ik\lambda}\left(\cos(\omega(\lambda)t)\widehat{q(0)}(\lambda)+\frac{\sin(\omega(\lambda)t)}{\omega(\lambda)}\widehat{p(0)}(\lambda)\right)d\lambda\longrightarrow0,\:t\rightarrow\infty.
\]
The equation $=^{(3)}$ uses the formulas (\ref{e_At_q_p}) and (\ref{g_Furie}).
And converging to zero takes place due to Corollary $2$ in \cite{KarazubaArchChub}
(p. $6$), since:
\begin{equation}
\left|\frac{1}{2\pi}\int_{0}^{2\pi}e^{ik\lambda}\left(\cos(\omega(\lambda)t)\widehat{q(0)}(\lambda)+\frac{\sin(\omega(\lambda)t)}{\omega(\lambda)}\widehat{p(0)}(\lambda)\right)d\lambda\right|\leqslant\frac{c}{\sqrt{t}},\label{delta_t_above}
\end{equation}
where the constant $c$ does not depend on $x$. Indeed, the corollary asserts:

\textit{Corollary 2 (Arkhipov, Karatsuba, Chubarikov).} Let $g(x)$
-- piecewise monotone continuous function, $p$ -- number of its monotonicity segments, 
$\max_{0\leqslant x\leqslant1}|g(x)|=H$. Let real-valued function $f(x)$ for $0<x<1$ has an $n$-order derivative, $n>1$,
moreover, for some $A>0$, for all $0<x<1$, the inequality
$|f^{(n)}(x)|\geqslant A$ holds. Then for $G=\int_{0}^{1}g(x)e^{2\pi if(x)}dx$
holds:
\begin{equation}
|G|\leqslant H\min\{1;\:24pnA^{-1/n}\}.\label{G_ineq_Karazuba}
\end{equation}

Let's make a change in the integral (\ref{delta_t_above}): 
\[
\frac{1}{2\pi}\int_{0}^{2\pi}e^{ik\lambda}\left(\cos(\omega(\lambda)t)\widehat{q(0)}(\lambda)+\frac{\sin(\omega(\lambda)t)}{\omega(\lambda)}\widehat{p(0)}(\lambda)\right)d\lambda=
\]
\[
=\int_{0}^{1}e^{ik2\pi\zeta}\left(\cos(\omega(2\pi\zeta)t)\widehat{q(0)}(2\pi\zeta)+\frac{\sin(\omega(2\pi\zeta)t)}{\omega(2\pi\zeta)}\widehat{p(0)}(2\pi\zeta)\right)d\zeta=
\]
\[
=\frac{1}{2}\int_{0}^{1}\cos\left(2\pi k\zeta\right)\widehat{q(0)}(2\pi\zeta)\left(e^{it\omega(2\pi\zeta)}+e^{-it\omega(2\pi\zeta)}\right)d\zeta+
\]
\[
+\frac{i}{2}\int_{0}^{1}\sin\left(2\pi k\zeta\right)\widehat{q(0)}(2\pi\zeta)\left(e^{it\omega(2\pi\zeta)}+e^{-it\omega(2\pi\zeta)}\right)d\zeta+
\]
\[
+\frac{1}{2i}\int_{0}^{1}\cos\left(2\pi k\zeta\right)\frac{\widehat{p(0)}(2\pi\zeta)}{\omega(2\pi\zeta)}\left(e^{it\omega(2\pi\zeta)}-e^{-it\omega(2\pi\zeta)}\right)d\zeta+
\]
\[
+\frac{1}{2}\int_{0}^{1}\sin\left(2\pi k\zeta\right)\frac{\widehat{p(0)}(2\pi\zeta)}{\omega(2\pi\zeta)}\left(e^{it\omega(2\pi\zeta)}-e^{-it\omega(2\pi\zeta)}\right)d\zeta,
\]
thus it is necessary to estimate $8$ integrals. Consider one
of them (for the rest the estimate is obtained similarly):

\[
\int_{0}^{1}\sin\left(2\pi k\zeta\right)\frac{\widehat{p(0)}(2\pi\zeta)}{\omega(2\pi\zeta)}e^{it\omega(2\pi\zeta)}d\zeta.
\]
According to the notation of the corollary $g(\zeta)=\sin\left(2\pi k\zeta\right)\frac{\widehat{p(0)}(2\pi\zeta)}{\omega(2\pi \zeta)}$.
Consider
\[
\omega^{2}(2\pi\zeta)=a(0)+2\sum_{k=1}^{K}a(k)\cos(2\pi k\zeta).
\]
As a linear combination of cosines it is continuous and piecewise monotonic
on $[0;1]$ function. The function $h(t)=\sqrt{t}$ is continuous and monotonic
over the entire domain of definition, whence we obtain that the composition $\omega(2\pi\zeta)$
of these two functions is continuous and piecewise monotonic on $[0;1]$.
Further, in view of (\ref{omega_2}), as well as the monotonicity and continuity of
function $h_{1}(t)=\frac{1}{t}$, function $\frac{1}{\omega(2\pi\zeta)}$
is continuous and piecewise monotonic as a composition of functions. Function $\widehat{p(0)}(2\pi\zeta)$
is also continuous as the Fourier transform of an element from $l_{2}\left(\mathbb{Z}\right)$,
hence $g(\zeta)$ is continuous and piecewise monotone as a product of
functions with corresponding properties. According to the Weierstrass theorem, on
$[0;1]$ it reaches its maximum on the interval, which we denote as $H$. Now consider $f(\zeta)=t\frac{\omega(2\pi\zeta)}{2\pi}$.
This function is $n$ times differentiable (we can take arbitrary $n>1$):
\[
f^{'}(\zeta)=t\omega'(2\pi\zeta)=-t\frac{\sum_{k=1}^{K}a(k)k\sin(2\pi k\zeta)}{\omega(2\pi\zeta)},
\]
\[
f^{''}(\zeta)=2\pi t\omega''(2\pi\zeta)=-2\pi t\frac{\left(\sum_{k=1}^{K}a(k)k^{2}\cos(2\pi k\zeta)\right)\omega(2\pi\zeta)-\omega'(2\pi\zeta)\sum_{k=1}^{K}a(k)k\sin(2\pi k\zeta)}{\omega^{2}(2\pi\zeta)}=
\]
\[
=-2\pi t\frac{\left(\sum_{k=1}^{K}a(k)k^{2}\cos(2\pi k\zeta)\right)\omega^{2}(2\pi\zeta)+\left(\sum_{k=1}^{K}a(k)k\sin(2\pi k\zeta)\right)^{2}}{\omega^{3}(2\pi\zeta)}.
\]
Numerator as a finite linear combination of trigonometric functions
has a finite number of zeros on the segment $[0;1]$ (if it has zeros somewhere at all). If there are no zeros, then the second derivative is separable
from zero, the required constant $A$ exists. If zeros exist,
then consider their $\varepsilon$-neighborhoods. Out of these neighborhoods
the second derivative is separable from zero. In these surroundings you can see
derivatives of order three or higher. In view of the analyticity of the function, and also
that it is not a constant, there exists the number $n$ of the derivative,
under which the corresponding derivative has no zeros in the chosen
neighborhood. In this case, in each neighborhood we have the estimate
(\ref{G_ineq_Karazuba}). Whence, in view of the fact that $A=A_{0}t$, by the corollary
we get the estimate (\ref{delta_t_above}). Which is what was required.

\subsubsection{Proof of the Theorem \ref{2_Theorem2}}

Strictly stationary process with finite first two
moments is a second-order stationary process, therefore,
we have the right to use the results of the previous Theorem. Then $\psi_{k}(t)-\eta_{k}(t)\overset{a.s.}{\rightarrow}0$.

The lemma is required: \begin{lemma} \label{2_lemma_Converge} Let
for $\xi(t),\,\eta(t)$ -- one-dimensional real random processes be
true: $\xi(t)-\eta(t)\overset{a.s.}{\rightarrow}0,\:t\rightarrow+\infty,$
and for any $t>0$ $\eta(t)\overset{d}{=}\eta_{0}$. Then $\xi(t)\overset{d}{\rightarrow}\eta_{0}$.
\end{lemma}

\begin{proof} $\xi(t)-\eta(t)\overset{a.s.}{\rightarrow}0$ implies
convergence $\xi(t)-\eta(t)\overset{P}{\rightarrow}0$. Further, from the condition
for all $t$ $\eta(t)\overset{d}{=}\eta_{0}$, convergence of $\eta(t)\overset{d}{\rightarrow}\eta_{0}$ follows. From where, according to the Slutsky lemma
for the function $g(x,y)=x+y$, we obtain the assertion of the lemma. \end{proof}

\subsubsection{Proof of the Theorem \ref{2_Theorem3}}

Let us represent the energy of the system in the following form:
\[
H(\psi)=\frac{1}{2}(\psi,G\psi),\;G=\left(\begin{array}{cc}
V & 0\\
0 & E
\end{array}\right).
\]
Let us represent the solution as
\[
\psi(t)=\epsilon(t)+\eta(t),
\]
\[
2H(\psi)=(\psi,G\psi)=(\psi,\psi)_{H}=(\epsilon,\epsilon)_{H}+(\epsilon,\eta)_{H}+(\eta,\epsilon)_{H}+(\eta,\eta)_{H}.
\]
Then the averages of the second and third terms are equal to zero due to the fact
that $E\eta=0$. Further,
\[
\eta^{T}G\eta=\int_{0}^{t}g^{T}e^{A^{T}(t-s_{1})}f(s_{1})ds_{1}\int_{0}^{t}Ge^{A(t-s_{2})}gf(s_{2})ds_{2}=
\]
\[
=\int_{0}^{t}\int_{0}^{t}(e^{A(t-s_{1})}g,e^{A(t-s_{2})}g)_{H}f(s_{1})f(s_{2})ds_{1}ds_{2}.
\]

Possibility of permutation of integration and taking the mathematical expectation
again follows from the continuity of the covariance function of the process $f(s)$
(from the condition of the theorem) and the existence of integrals
\begin{equation}
\int_{0}^{t}f(s)Ge^{A(ts)}gds,\quad\int_{0}^{t}f(s)e^{A(ts)}gds,\label {scal_H_int}
\end{equation}
(in the mean square sense) and the Riemann integral existence
\[
\int_{0}^{t}\int_{0}^{t}B(s-\tau)g^{T}e^{A^{T}(t-s)}Ge^{A(t-\tau)}gdsd\tau.
\]
The last integral exists due to integrands continuity and boundedness, and for the existence of integrals (\ref{scal_H_int})
already described conditions are enough (see \cite{KramerLeadbetter} p.
$94-129$). Then
\begin{equation}
E(\eta,\eta)_{H}=\int_{0}^{t}\int_{0}^{t}(e^{A(t-s_{1})}g,e^{A(t-s_{2})}g)_{H}B(s_{1}-s_{2})ds_{1}ds_{2}.\label{E_H_t}
\end{equation}
We introduce the scalar product by the formula
\[
(\psi_{1},\psi_{2})_{H}=\frac{1}{2\pi}\int_{0}^{2\pi}\widehat{p_{1}}(\lambda)\overline{\widehat{p_{2}}}(\lambda)+\omega^{2}(\lambda)\widehat{q_{1}}(\lambda)\overline{\widehat{q_{2}}}(\lambda)d\lambda.
\]
Then

\[
(\psi,\psi)_{H}=\frac{1}{2\pi}\int_{0}^{2\pi}|\widehat{p}(\lambda)|^{2}+\omega^{2}(\lambda)|\widehat{q}(\lambda)|^{2}d\lambda.
\]

Further, taking into account (\ref{exp_At}) we get:

\begin{equation}
\widehat{e^{At}\left(\begin{array}{c}
q\\
p
\end{array}\right)}=\left(\begin{array}{c}
\cos(\omega(\lambda)t)\widehat{q}(\lambda)+\frac{\sin(\omega(\lambda)t)}{\omega(\lambda)}\widehat{p}(\lambda)\\
-\omega(\lambda)\sin(\omega(\lambda)t)\widehat{q}(\lambda)+\cos(\omega(\lambda)t)\widehat{p}(\lambda)
\end{array}\right),\label{e_At_q_p}
\end{equation}
and, since
\begin{equation}
\widehat{g}=\left(\begin{array}{c}
\widehat{0}\\
\widehat{e_{n}}
\end{array}\right)=\left(\begin{array}{c}
0\\
e^{in\lambda}
\end{array}\right),\label{g_Furie}
\end{equation}
so

\[
\widehat{e^{At}g}=e^{in\lambda}\left(\begin{array}{c}
\frac{\sin t\omega(\lambda)}{\omega(\lambda)}\\
\cos t\omega(\lambda)
\end{array}\right),
\]
whence from (\ref{E_H_t}) follows:
\[
E(\eta,\eta)_{H}=\frac{1}{2\pi}\int_{0}^{t}\int_{0}^{t}\int_{0}^{2\pi}\cos((t-s_{1})\omega(\lambda))\cos((t-s_{2})\omega(\lambda))B(s_{1}-s_{2})d\lambda ds_{1}ds_{2}+
\]
\[
+\frac{1}{2\pi}\int_{0}^{t}\int_{0}^{t}\int_{0}^{2\pi}\sin((t-s_{1})\omega(\lambda))\sin((t-s_{2})\omega(\lambda))B(s_{1}-s_{2})d\lambda ds_{1}ds_{2}=
\]
\[
=\frac{1}{2\pi}\int_{0}^{t}\int_{0}^{t}\int_{0}^{2\pi}\cos((s_{1}-s_{2})\omega(\lambda))B(s_{1}-s_{2})d\lambda ds_{1}ds_{2}.
\]
Further, in view of (\ref{process}):
\[
E(\eta,\eta)_{H}=\frac{1}{2\pi}\int_{0}^{t}\int_{0}^{t}\int_{0}^{2\pi}\cos((s_{1}-s_{2})\omega(\lambda))\int_{\mathbb{R}}e^{ix(s_{1}-s_{2})}\mu(dx)d\lambda ds_{1}ds_{2}=
\]
\[
=\frac{1}{2\pi}\int_{\mathbb{R}}\int_{0}^{2\pi}f(t,\lambda,x)d\lambda\:\mu(dx),
\]
where
\begin{equation}
f(t,\lambda,x)=\int_{0}^{t}\int_{0}^{t}\cos((s_{1}-s_{2})\omega(\lambda))e^{ix(s_{1}-s_{2})}ds_{1}ds_{2}=\frac{1}{2}(f_{+}(t,\lambda,x)+f_{-}(t,\lambda,x)),\label{f_t_lambda_x}
\end{equation}
where
\[
f_{\pm}(t,\lambda,x)=\int_{0}^{t}\int_{0}^{t}e^{i(s_{1}-s_{2})(x\pm\omega(\lambda))}ds_{1}ds_{2}.
\]
Denote $\gamma=x\pm\omega(\lambda)$, then
\[
f_{\pm}(t,\lambda,x)=\int_{0}^{t}\int_{0}^{t}e^{i(s_{1}-s_{2})\gamma}ds_{1}ds_{2}=\left|\int_{0}^{t}e^{is\gamma}ds\right|^{2}=\left|\frac{e^{is\gamma}}{i\gamma}|_{0}^{t}\right|^{2}=\frac{1}{\gamma^{2}}\left|e^{it\gamma}-1\right|^{2}=
\]
\[
=\frac{1}{\gamma^{2}}\left(\left(1-\cos t\gamma\right)^{2}+\left(\sin t\gamma\right)^{2}\right)=\frac{2}{\gamma^{2}}\left(1-\cos t\gamma\right).
\]
Then (\ref{f_t_lambda_x}) is converted to:
\[
f(t,\lambda,x)=\frac{1-\cos t(x+\omega(\lambda))}{(x+\omega(\lambda))^{2}}+\frac{1-\cos t(x-\omega(\lambda))}{(x-\omega(\lambda))^{2}}=I_{+}(t,\lambda,x)+I_{-}(t,\lambda,x),
\]
further
\[
E(\eta,\eta)_{H}=\frac{1}{2\pi}\int_{\mathbb{R}}\int_{0}^{2\pi}\left(I_{+}(t,\lambda,x)+I_{-}(t,\lambda,x)\right)d\lambda\:\mu(dx),
\]
\[
\int_{0}^{2\pi}I_{+}(t,\lambda,x)d\lambda=\int_{0}^{2\pi}\frac{1-\cos t(x+\omega(\lambda))}{(x+\omega(\lambda))^{2}}d\lambda=\int_{0}^{2\pi}\frac{d\lambda}{(x+\omega(\lambda))^{2}}-\int_{0}^{2\pi}\frac{\cos t(x+\omega(\lambda))}{(x+\omega(\lambda))^{2}}d\lambda.
\]

By corollary $2$ in \cite{KarazubaArchChub} (p. $6$) has
place estimation:
\begin{equation}
r_{+}(t,x)=\left|\int_{0}^{2\pi}\frac{\cos t(x+\omega(\lambda))}{(x+\omega(\lambda))^{2}}d\lambda\right|\leqslant\frac{c}{\sqrt{t}},\label{r_+}
\end{equation}
where the constant $c$ does not depend on $x$. Indeed, we make the substitution
in the integral in (\ref{r_+}):
\[
\int_{0}^{2\pi}\frac{\cos t(x+\omega(\lambda))}{(x+\omega(\lambda))^{2}}d\lambda=2\pi\int_{0}^{1}\frac{\cos t(x+\omega(2\pi\zeta))}{(x+\omega(2\pi\zeta))^{2}}d\zeta=
\]
\[
=\pi\int_{0}^{1}\frac{e^{it(x+\omega(2\pi\zeta))}+e^{-it(x+\omega(2\pi\zeta))}}{(x+\omega(2\pi\zeta))^{2}}d\zeta.
\]
According to the notation of the corollary $g(\zeta)=\frac{1}{(x+\omega(2\pi\zeta))^{2}}>0$.
From the proof of item $2$ of Theorem \ref{2_Theorem1} we have that
$\omega(2\pi\zeta)$ is continuous and piecewise monotonic on $[0;1]$. Further,
the support of the measure $\mu$ is separated from the root of the spectral set of the operator
$V$, and therefore the function $x+\omega(2\pi\zeta)$ is separated from zero for
all $\zeta\in[0;1]$. Then $(x+\omega(\lambda))^{2}$ is separated
from zero and, in view of the monotonicity and continuity of the function $h_{1}(t)=\frac{1}{t}$,
$g(\zeta)$ is continuous and piecewise monotonic as a composition of functions.
By the Weierstrass theorem it reaches its maximum on $[0;1]$,
which we denote by $H$. Now consider $f(\zeta)=\frac{t(x+\omega(2\pi\zeta))}{2\pi}$.
This function is $n$ times differentiable (we can take arbitrary $n>1$):
\[
f^{'}(\zeta)=t\omega'(2\pi\zeta)=-t\frac{\sum_{k=1}^{K}a(k)k\sin(2\pi k\zeta)}{\omega(2\pi\zeta)},
\]
\[
f^{''}(\zeta)=2\pi t\omega''(2\pi\zeta)=-2\pi t\frac{\left(\sum_{k=1}^{K}a(k)k^{2}\cos(2\pi k\zeta)\right)\omega(2\pi\zeta)-\omega'(2\pi\zeta)\sum_{k=1}^{K}a(k)k\sin(2\pi k\zeta)}{\omega^{2}(2\pi\zeta)}=
\]
\[
=-2\pi t\frac{\left(\sum_{k=1}^{K}a(k)k^{2}\cos(2\pi k\zeta)\right)\omega^{2}(2\pi\zeta)+\left(\sum_{k=1}^{K}a(k)k\sin(2\pi k\zeta)\right)^{2}}{\omega^{3}(2\pi\zeta)}.
\]
Numerator as a finite linear combination of trigonometric functions
has a finite number of zeros on the segment $[0;1]$ (if it has
zeros somewhere). If there are no zeros, then the second derivative is separable from
zero, the required constant $A$ exists. If zeros exist,
then consider their $\varepsilon$-neighborhoods. Out of these neighborhoods
the second derivative is separable from zero. In these surroundings you can see
derivatives of order three or higher. In view of the analyticity of the function, and also
that it is not a constant, there is an index $n$ of the derivative,
which has no zeros in the chosen neighborhood. In that case, in each
neighborhood, the estimate (\ref{G_ineq_Karazuba}) takes place. From where in view
$A=A_{0}t$ we obtain the estimate (\ref{r_+}) by the corollary.
Which is what was required.

Hereof 
\[
E(\eta,\eta)_{H}=\frac{1}{2\pi}\int_{\mathbb{R}}\int_{0}^{2\pi}\left(\frac{1}{(x+\omega(\lambda))^{2}}+\frac{1}{(x-\omega(\lambda))^{2}}\right)d\lambda\:\mu(dx)+R(t),
\]
\[
|R(t)|\leqslant\frac{1}{2\pi}\int_{\mathbb{R}}\frac{2c}{\sqrt{t}}\mu(dx)=\frac{c}{\pi\sqrt{t}}B(0)=\frac{c\sigma^{2}}{\pi\sqrt{t}},\;\sigma^{2}=Ef_{s}^{2}.
\]

Where do we get that
\[
\lim_{t\rightarrow+\infty}E(\eta,\eta)_{H}=\frac{1}{4\pi}\int_{\mathbb{R}}\int_{0}^{2\pi}\left(\frac{1}{(x+\omega(\lambda))^{2}}+\frac{1}{(x-\omega(\lambda))^{2}}\right)d\lambda\:\mu(dx)=
\]
\[
=\frac{1}{2\pi}\int_{\mathbb{R}}\int_{0}^{2\pi}\frac{\omega^{2}(\lambda)+x^{2}}{(\omega^{2}(\lambda)-x^{2})^{2}}d\lambda\mu(dx).
\]
It remains to find $E(\epsilon,\epsilon)_{H}$:
\begin{equation}
(\epsilon,\epsilon)_{H}=\frac{1}{2\pi}\int_{0}^{2\pi}|\widehat{p_{\epsilon}}(\lambda)|^{2}+\omega^{2}(\lambda)|\widehat{q_{\epsilon}}(\lambda)|^{2}d\lambda,\label{scalProddelta}
\end{equation}
where
\[
q_{\epsilon}=\cos(\omega(\lambda)t)\widehat{q(0)}(\lambda)+\frac{\sin(\omega(\lambda)t)}{\omega(\lambda)}\widehat{p(0)}(\lambda),
\]
\[
p_{\epsilon}=-\omega(\lambda)\sin(\omega(\lambda)t)\widehat{q(0)}(\lambda)+\cos(\omega(\lambda)t)\widehat{p(0)}(\lambda),
\]
what follows from (\ref{e_At_q_p}). Hence: 
\[
|\widehat{p_{\epsilon}}(\lambda)|^{2}+\omega^{2}(\lambda)|\widehat{q_{\epsilon}}(\lambda)|^{2}=(-\omega(\lambda)\sin(\omega(\lambda)t)\widehat{q(0)}(\lambda)+
\]
\[
+\cos(\omega(\lambda)t)\widehat{p(0)}(\lambda))(-\omega(\lambda)\sin(\omega(\lambda)t)\overline{\widehat{q(0)}}(\lambda)+\cos(\omega(\lambda)t)\overline{\widehat{p(0)}}(\lambda))+
\]
\[
+\omega^{2}(\lambda)(\cos(\omega(\lambda)t)\widehat{q(0)}(\lambda)+\frac{\sin(\omega(\lambda)t)}{\omega(\lambda)}\widehat{p(0)}(\lambda))(\cos(\omega(\lambda)t)\overline{\widehat{q(0)}}(\lambda)+
\]
\[
+\frac{\sin(\omega(\lambda)t)}{\omega(\lambda)}\overline{\widehat{p(0)}}(\lambda))=|\widehat{p(0)}(\lambda)|^{2}+\omega^{2}(\lambda)|\widehat{q(0)}(\lambda)|^{2},
\]
which in view of (\ref{scalProddelta}) leads to
\[
(\epsilon,\epsilon)_{H}=\frac{1}{2\pi}\int_{0}^{2\pi}|\widehat{p(0)}(\lambda)|^{2 }+\omega^{2}(\lambda)|\widehat{q(0)}(\lambda)|^{2}d\lambda=2H(\psi(0)),
\]
whence, in view of the non-randomness of the last expression:
\[
E(\epsilon,\epsilon)_{H}=2H(\psi(0)).
\]

Now we find the average energy of the limiting distribution:
\[
\eta(t)=(q^{\infty}(t),\,p^{\infty}(t)).
\]
Since this vector is a stationary solution of the system under study
(and the only one with such initial conditions) and the Fourier transform
is a linear operator then the Fourier transform of this vector
will go into the stationary solution of the resulting equation:

\begin{equation}
\ddot{Q}_{\lambda}(t)=-\omega^{2}Q_{\lambda}(t)+f_{t}e^{in\lambda},\qquad\omega=\omega(\lambda),\label{eq_Furie}
\end{equation}
\begin{equation}
Q_{\lambda}(0)=\widehat{q^{\infty}(0)}(\lambda),\;\dot{Q}_{\lambda}(0)=\widehat{p^{\infty}(0)}(\lambda).\label{cond_Furie}
\end{equation}
Functions given by formulas:

\[
\widehat{q^{\infty}(t)}(\lambda)=\int_{\mathbb{R}}\frac{e^{itx}e^{in\lambda}}{\omega^{2} (\lambda)-x^{2}}Z(dx),
\]
\[
\widehat{p^{\infty}(t)}(\lambda)=\frac{d}{dt}\widehat{q^{\infty}(t)}(\lambda)=\int_{\mathbb{ R}}\frac{ixe^{itx}e^{in\lambda}}{\omega^{2}(\lambda)-x^{2}}Z(dx),
\]
give stationary solution of the equation with the corresponding initial
conditions. Indeed, direct verification shows that this is the solution of
equations (\ref{eq_Furie}). It remains to prove that the initial
conditions also correspond to it. To do this, we find the Fourier transform $\eta(0)$:
\[
\eta(0)=-\int_{\mathbb{R}}R_{A}(ix)gZ(dx),
\]
\[
\widehat{\eta(0)}(\lambda)=-\int_{\mathbb{R}}\widehat{R_{A}(ix)g}Z(dx),
\]
where the series summation and integration are interchanged
which is possible due to the convergence of the integral $\int_{\mathbb{R}}\left|R_{A}(ix)g\right|^{2}\mu(dx).$

Let us find $\widehat{R_{A}(ix)g}$. From (\ref{R_A_ix_g}) we have:
\[
\widehat{R_{A}(ix)g}=\widehat{\left(\begin{array}{c}
-R_{V}(x^{2})e_{n}\\
-ixR_{V}(x^{2})e_{n}
\end{array}\right)}.
\]
Because
\[
\widehat{R_{V}(x^{2})e_{n}}=\frac{e^{in\lambda}}{\omega^{2}(\lambda)-x^{2}},
\]
(see \cite{BogachevSmolyanov}, pp. $357-359$), we arrive at
\[
\widehat{R_{A}(ix)g}=-\frac{e^{in\lambda}}{\omega^{2}(\lambda)-x^{2}}\left(\begin{array}{c}
1\\
ix
\end{array}\right),
\]
i.e. to (\ref{cond_Furie}). Which is what was required.

Then: 
\[
EH(\eta(0))=\frac{1}{4\pi}\int_{\mathbb{R}}\int_{0}^{2\pi}\left(\frac{x^{2}}{(\omega^{2}(\lambda)-x^{2})^{2}}+\frac{\omega^{2}(\lambda)}{(\omega^{2}(\lambda)-x^{2})^{2}}\right)d\lambda\:\mu(dx)=
\]
\[
=\frac{1}{4\pi}\int_{\mathbb{R}}\int_{0}^{2\pi}\frac{\omega^{2}(\lambda)+x^{2}}{(\omega^{2}(\lambda)-x^{2})^{2}}d\lambda\mu(dx)=\alpha.
\]
The theorem is proved completely.

\subsection{Acknowledgment}

We would like to thank professor Vadim Malyshev for stimulating discussions
and numerous remarks.

\end{document}